\documentclass[10pt]{article}

\providecommand{\UsePackage}[1]{%
   \IfFileExists{#1.sty}%
   {\usepackage{#1}}%
   {\usepackage{styles/#1}}%
}

\UsePackage{salgorithm}%
\usepackage{lmodern}
\usepackage{color}%
\usepackage{xspace}%
\usepackage{euscript}%
\usepackage{amstext}
\usepackage{paralist}%
\usepackage{amsmath}%
\usepackage{amssymb}%
\usepackage{graphicx}
\usepackage[hyperref,amsmath,thmmarks]{ntheorem}
\theoremseparator{.}
\usepackage{titlesec}
\titlelabel{\thetitle. }

\usepackage[pdfpagelabels,plainpages=false,hypertexnames=false]{hyperref}%
\hypersetup{%
   breaklinks,%
   ocgcolorlinks,
  colorlinks=true,%
  linkcolor=[rgb]{0.45,0.0,0.0},%
  urlcolor=[rgb]{0.05,0.390,0.0},%
  citecolor=[rgb]{0,0,0.45}
}
\usepackage[all]{hypcap}

\newcommand{\remove}[1]{}

\numberwithin{figure}{section}%
\numberwithin{table}{section}%
\numberwithin{equation}{section}%

\definecolor{blue25}{rgb}{0,0,0.7}
\newcommand{\emphic}[2]{%
     \textcolor{blue25}{%
         \textbf{\emph{#1}}}%
         \index{#2}}
\newcommand{\emphi}[1]{\emphic{#1}{#1}}

\newcommand{\obslab}[1]{\label{observation:#1}}
\newcommand{\obsref}[1]{Observation~\ref{observation:#1}}

\providecommand{\dslab}[1]{\label{d:s:#1}}
\providecommand{\dsref}[1]{\ensuremath{\mathcal{DS}_{\ref{d:s:#1}}}}

\providecommand{\lemlab}[1]{\label{lemma:#1}}
\providecommand{\lemref}[1]{Lemma~\ref{lemma:#1}}

\providecommand{\corlab}[1]{\label{corollary:#1}}
\providecommand{\corref}[1]{Corollary~\ref{corollary:#1}}

\newtheorem{theorem}{Theorem}[section] 
\newtheorem{lemma}[theorem]{Lemma}

\newtheorem{corollary}[theorem]{Corollary}
{\theorembodyfont{\rm} \newtheorem{definition}[theorem]{Definition}}
{\theorembodyfont{\rm} }
\newtheorem{observation}[theorem]{Observation}
\newtheorem{assumption}[theorem]{Assumption}
\newtheorem{datastructure}[theorem]{Data-structure}

\newenvironment{proof}{\trivlist\item[]\emph{Proof}:}%
                  {\unskip\nobreak\hskip 1em plus 1fil\nobreak%
                           \rule{2mm}{2mm}
                           \parfillskip=0pt%
                           \endtrivlist}

\newcommand{\Term}[1]{\textsf{#1}}
\newcommand{\TermI}[1]{\Term{#1}\index{#1@\Term{#1}}}
\newcommand{\BFS}{\textsf{BFS}\ }

\newcommand{\thmlab}[1]{{\label{theo:#1}}}
\newcommand{\thmref}[1]{Theorem~\ref{theo:#1}}

\newcommand{\eqlab}[1]{\label{equation:#1}}

\newcommand{\Eqrefpage}[1]{Eq.~(\ref{equation:#1})%
   $_\text{p\pageref{equation:#1}}$}
\newcommand{\seclab}[1]{{\label{section:#1}}}
\newcommand{\secref}[1]{Section~\ref{section:#1}}

\newcommand{\deflab}[1]{\label{defn:#1}}
\newcommand{\defref}[1]{Defnition~\ref{defn:#1}}

\newcommand{\assumplab}[1]{\label{assumption:#1}}
\newcommand{\assumpref}[1]{Assumption~\ref{assumption:#1}}

\newcommand{\itemlab}[1]{\label{item:#1}}
\newcommand{\itemref}[1]{(\ref{item:#1})}

\newcommand{\MakeBig}{\rule[-.2cm]{0cm}{0.4cm}}
\newcommand{\MakeSBig}{\rule[0.0cm]{0.0cm}{0.35cm}} 
\newcommand{\brc}[1]{\left\{ {#1} \right\}}
\newcommand{\sep}[1]{\,\left|\, {#1} \MakeBig\right.}
\newcommand{\pth}[2][\!]{#1\left({#2}\right)}
\newcommand{\pbrc}[2][\!\!]{#1\left[ {#2} \MakeBig \right]}

\newcommand{\norm}[1]{\left\lVert {#1} \right \rVert}

\newcommand{\distPk}[3]{\mathsf{d}_{#3}\pth{#2,#1}}

\newcommand{\abs}[1]{\left | {#1} \right |}
\newcommand{\floor}[1]{\left\lfloor {#1} \right\rfloor}

\newcommand{\cardin}[1]{\left\lvert {#1} \right\rvert}
\newcommand{\ceiling}[1]{\lceil #1 \rceil}

\newcommand{\eps}{{\varepsilon}}%

\newcommand{\WSPD}{\TermI{WSPD}\xspace}
\newcommand{\NN}{\TermI{NN}\xspace}
\newcommand{\ANN}{\TermI{ANN}\xspace}
\newcommand{\AVD}{\TermI{AVD}\xspace}
\newcommand{\etal}{\textit{et~al.}\xspace}
\renewcommand{\Re}{{\rm I\!\hspace{-0.025em} R}}

\newcommand{\diameter}[1]{\mathsf{diam}\pth{ {#1} }}
\newcommand{\PntSet}{\mathsf{P}}
\newcommand{\PntSetQ}{\mathsf{Q}}

\newcommand{\query}{\mathtt{q}}

\newcommand{\pnt} {\mathsf{p}}
\newcommand{\pntA}{\mathsf{u}}
\newcommand{\pntB}{\mathsf{v}}

\providecommand{\si}[1]{#1}

\providecommand{\ds}{\displaystyle}

\newcommand{\ball}[2]{\mathsf{b}(#1,#2)}

\newcommand{\setA}{X}
\newcommand{\setB}{Y}

\newcommand{\ctrA}{\mathsf{c}}
\newcommand{\ctrC}{\mathsf{w}}
\newcommand{\ctrD}{\mathsf{u}}

\newcommand{\ballA}{b}
\newcommand{\ballB}{b'}
\newcommand{\ballC}{\Lambda}
\newcommand{\ballD}{\mathsf{o}}

\newcommand{\ballQ}{\ballA_{\query}}

\newcommand{\radA}{\mathsf{r}}
\newcommand{\radC}{\mathsf{x}}
\newcommand{\radD}{\psi}%

\newcommand{\CenterSetA}{\mathcal{C}}
\newcommand{\dist}[2]{\mathsf{d}\pth{#1,#2}}
\newcommand{\distE}[2]{\left\|{#1 - #2}\right\|}

\newcommand{\NodeSetA}{X_N}

\newcommand{\BallSetA}{\mathcal{B}}
\newcommand{\BallSetX}{\mathcal{R}}

\newcommand{\GridACinner}{\Grid_\approx}
\newcommand{\GridAC}{\GridACinner}
\newcommand{\GridApproxX}[2]{\GridACinner\pth{#1, #2}}
\newcommand{\GridSetB}[1]{\GridACinner\pth{#1}}

\newcommand{\IntBallsX}[2]{{#1}_\succeq\pth{#2}}

\newcommand{\Env}{G}

\newcommand{\constA}{\zeta_1}
\newcommand{\cellA}{\mathsf{\Box}}
\newcommand{\QTree}{\mathcal{T}}
\newcommand{\QTreeCen}{\QTree_\CenterSetA}

\newcommand{\node}{\nu}
\newcommand{\nodeA}{u}
\newcommand{\DS}{\mathcal{D}}

\newcommand{\num}{x}
\newcommand{\numA}{\gamma}
\newcommand{\numB}{\zeta}
\newcommand{\kdist}{d_k}

\newcommand{\BSet}[2]{\capSubset{\BallSetA}{\ball{#1}{#2}}}

\newcommand{\BSetNum}[2]{\cardin{\BSet{#1}{#2}}}
\newcommand{\algBSet}{\Algorithm{range{}Count}\ }
\newcommand{\algBSetX}[3]{\algBSet{}%
   {\ensuremath{\pth[]{#1,#3, #2}}}}

\newcommand{\NxNum}{N}
\newcommand{\NxNumX}[1]{\NxNum\pth{#1}}
\newcommand{\CxNum}{C}
\newcommand{\cval}{x}

\newcommand{\CellSetA}{\mathcal{I}}
\renewcommand{\th}{\si{th}\xspace}
\newcommand{\bcountI}{\#}
\newcommand{\bcountX}[1]{\bcountI\pth{#1}}

\newcommand{\Grid}{\mathsf{G}\index{grid}}

\newcommand{\diamX}[2][\!]{\mathrm{diam}\pth[#1]{#2}}

\newcommand{\BallsAssocX}[1]{\BallSetA_{#1}}
\newcommand{\ApproxX}[1]{\widetilde{#1}}

\newcommand{\capSubset}[2]{{#1}\pth{#2}}

\newcommand{\constC}{\xi}
\newcommand{\nclust}{m}
\newcommand{\mapped}[1]{#1'}
\newcommand{\Scube}{U}
\newcommand{\repX}[1]{\mathrm{nn'}\pth{#1}}
\newcommand{\CellSetAVD}{\mathcal{W}}
\newcommand{\CellSetC}{\mathcal{S}}
\newcommand{\repres}[1]{#1_{\mathsf{rep}}}
\newcommand{\ballrepX}[1]{\mathbf{b}\pth{#1}}
\newcommand{\adknn}[1]{\mathrm{\beta}_k\pth{#1}}
\newcommand{\kbannrepX}[1]{\mathrm{nn}_k\pth{#1}}
\newcommand{\normP}[1]{\norm{#1}_{\oplus}}

\newcommand{\cDim}{\mathsf{c}_d}
\newcommand{\QClustBalls}{\Sigma}
\newcommand{\estX}{\lambda^*}

\newcommand{\minDistPk}[2]{r_{\mathrm{opt}}\pth{{#1},{#2}}}
\newcommand{\minBCoverk}[2]{R_{\mathrm{opt}}\pth{{#1},{#2}}}

\newcommand{\SarielThanks}[1]{\thanks{Department of Computer
      Science; 
      University of Illinois; 
      201 N. Goodwin Avenue;
      Urbana, IL, 61801, USA;
      {\tt sariel\atgen{}uiuc.edu}; {\tt
         \url{http://www.uiuc.edu/\string~sariel/}.} #1}}
\newcommand{\NirmanThanks}[1]{\thanks{Department of Computer
      Science; 
      University of Illinois; 
      201 N. Goodwin Avenue;
      Urbana, IL, 61801, USA;
      {\tt \si{nkumar5}\atgen{}illinois.edu}; {\tt
         \url{http://www.cs.uiuc.edu/\string~\si{nkumar5}/}.} #1}}
\newcommand{\atgen}{\symbol{'100}}

\title{Robust Proximity Search for Balls using Sublinear Space%
       \footnote{Work on this paper was partially support by NSF AF awards
       CCF-0915984 and CCF-1217462}}

\author{Sariel Har-Peled%
   \SarielThanks{}%
   \and%
   Nirman Kumar%
   \NirmanThanks{}%
}

\begin{document}

\maketitle

\begin{abstract}
    Given a set of $n$ disjoint balls $\ballA_1, \dots, \ballA_n$ in
    $\Re^d$, we provide a data structure, of near linear size, that
    can answer $(1\pm\eps)$-approximate $k$\th-nearest neighbor
    queries in $O(\log n + 1/\eps^d)$ time, where $k$ and $\eps$ are
    provided at query time. If $k$ and $\eps$ are provided in advance,
    we provide a data structure to answer such queries, that requires
    (roughly) $O(n/k)$ space; that is, the data structure has
    sublinear space requirement if $k$ is sufficiently large.
\end{abstract}

\section{Introduction}
The \emph{nearest neighbor} problem is a fundamental problem in
Computer Science \cite{sdi-nnmlv-06, ai-nohaa-08}. 
Here, one is given a set of points $\PntSet$, and
given a query point $\query$ one needs to output the nearest point in
$\PntSet$ to $\query$. There is a trivial $O(n)$ algorithm for
this problem. Typically the set of data points is fixed, while
different queries keep arriving. Thus, one can use preprocessing to
facilitate a faster query. There are several applications of
nearest neighbor search in computer science including 
pattern recognition, information retrieval, vector
compression, computational statistics, clustering, data mining
and learning among many others, see for instance 
the survey by Clarkson \cite{c-nnsms-06} for references.
If one is
interested in guaranteed performance and near linear space, there is
no known way to solve this problem efficiently (i.e., logarithmic
query time) for dimension $d > 2$, while using near linear space for
the data structure.


In light of the above, major effort has been devoted to develop
approximation algorithms for nearest neighbor search
\cite{amnsw-oaann-98, im-anntr-98, 
   \si{c-nnsms-06}, \si{him-anntr-12}}. 
In
the \emphi{$(1+\eps)$-approximate nearest neighbor} problem, one is
additionally given an approximation parameter $\eps > 0$ and one is
required to find a point $\pntA \in \PntSet$ such that
$\dist{\query}{\pntA} \leq (1+\eps) \dist{\query}{\PntSet}$.  In $d$
dimensional Euclidean space, one can answer \ANN queries in $O(\log n
+ 1/\eps^{d-1})$ time using linear space \cite{amnsw-oaann-98,
   h-gaa-11}. Unfortunately, the constant hidden in the $O$ notation
is exponential in the dimension (and this is true for all bounds
mentioned in this paper), and specifically because of the
$1/\eps^{d-1}$ in the query time, this approach is only efficient in
low dimensions. Interestingly, for this data structure, the
approximation parameter $\eps$ need not be specified during the
construction, and one can provide it during the query. An alternative
approach is to use Approximate Voronoi Diagrams (\AVD), introduced by
Har-Peled \cite{h-rvdnl-01}, which is a partition of space into
regions of low total complexity, with a representative point for each
region, that is an \ANN for any point in the region. In particular,
Har-Peled showed that there is such a decomposition of size $O\pth{(n
   /\eps^d)\log^2 n}$, see also \cite{him-anntr-12}.  This allows \ANN
queries to be answered in $O( \log n)$ time.  Arya and Malamatos
\cite{am-lsavd-02} showed how to build \AVD{}s of linear complexity
(i.e., $O(n/\eps^d)$). Their construction uses \WSPD (Well Separated
Pairs Decomposition) \cite{ck-dmpsa-95}. Further trade-offs between
query time and space usage for \AVD{}s were studied by Arya \etal
\cite{amm-sttan-09}.

A more general problem is the $k$-nearest neighbors problem where one
is interested in finding the $k$ points in $\PntSet$ nearest to the
query point $\query$. This is widely used in classification, where the
majority label is used to label the query point. A restricted version
is to find only the $k$\th-nearest neighbor. This problem and its
approximate version have been considered in \cite{amm-sttas-05,
   hk-drhrp-12}. 

Recently, the authors \cite{hk-drhrp-12} showed that one can compute
a $(k,\eps)$-\AVD that $(1+\eps)$-approximates the distance to the
$k$\th nearest neighbor, and surprisingly, requires $O(n/k)$
space; that is, sublinear space if $k$ is sufficiently large. For
example, for the case $k = \Omega( \sqrt{n})$, which is of interest in
practice, the space required is only $O\pth{\sqrt{n}}$.  
Such
\ANN is of interest when one is worried that there is noise in the
data, and thus one is interested in the distance to the $k$\th \NN
which is more robust and noise resistant. Alternatively,
one can think about such data structures as enabling one to summarize
the data in a way that still facilitates meaningful proximity queries.

In this paper we consider a generalization of the $k$\th-nearest
neighbor problem. Here, we are given a set of $n$ disjoint balls in
$\Re^d$ and we want to preprocess them, so that given a query point we
can find approximately the $k$\th closest ball. The distance of a
query point to a ball is defined as the distance to its boundary if
the point is outside the ball or $0$ otherwise. Clearly, this problem
is a generalization of the $k$\th-nearest neighbor problem by viewing
points as balls of radius $0$.  Algorithms for the $k$\th-nearest
neighbor for points, do not extend in a straightforward manner to this
problem because the distance function is no longer a metric. Indeed,
there can be two very far off points both very close to a single ball,
and thus the triangle inequality does not hold.  The problem of
finding the closest ball can also be modeled as a problem of
approximating the minimization diagram of a set of functions; here, a
function would correspond to the distance from one of the given
balls. There has been some recent work by the authors on this topic,
see \cite{hk-amdgp-13}, where a fairly general class of functions
admits a near-linear sized data structure permitting a logarithmic
time query for the problem of approximating the minimization
diagram. However, the problem that we consider in this paper does not
fall under the framework of \cite{hk-amdgp-13}. The
technical assumptions of \cite{hk-amdgp-13} mandate that the set of points
which form the $0$-sublevel set of a distance function, i.e., the set
of points at which the distance function is $0$ is a single point (or 
an empty set). This is not the case for the problem we consider here. 
Also, we are interested in the more general
$k$\th-nearest neighbor problem, while \cite{hk-amdgp-13} only considers
the nearest-neighbor problem, i.e., $k = 1$.

We first show how to preprocess the set of balls into a
data structure requiring space $O(n)$, in $O(n \log n)$ time, so that
given a query point $\query$, a number $1 \leq k \leq n$ and $\eps >
0$, one can compute a $(1 \pm \eps)$-approximate $k$\th closest ball
in time $O(\log n + \eps^{-d})$. If both $k$ and $\eps$ are available
during preprocessing, one can preprocess the balls into a
$(k,\eps)$-\AVD, using $O(\frac{n}{k \eps^d}\log(1/\eps))$ space, so
that given a query point $\query$, a $(k,\eps)$-\ANN closest ball can
be computed, in $O(\log (n/k) + \log (1/\eps))$ time.

\section*{Paper Organization}
In \secref{prelims}, we define the problem, list some assumptions, and
introduce notations.  In \secref{apprx:rcount}, we set up some basic
data structures to answer approximate range counting queries for
balls. In \secref{k:at:query}, we present the data structure, query
algorithm and proof of correctness for our data structure which can
compute $(1 \pm \eps)$-approximate $k$\th-nearest neighbors of a query
point when $k, \eps$ are only provided during query time. In \secref{quorum}
we present approximate quorum clustering, see \cite{cdhks-gqsa-05,hk-drhrp-12}, 
for a set of disjoint balls. Using
this, in
\secref{avd}, we present the $(k,\eps)$-\AVD construction. We conclude
in \secref{conclusions}.

\section{Problem definition and notation}
\seclab{prelims}

We are given a set of disjoint\footnote{Our data structure and
   algorithm work for the more general case where the balls are
   interior disjoint, where we define the interior of a ``point
   ball'', i.e., a ball of radius $0$, as the point itself. This is not
   the usual topological definition.} balls $\BallSetA = \brc{
   \ballA_1, \dots, \ballA_n }$, where $\ballA_i =
\ball{\ctrA_i}{\radA_i}$, for $i=1, \ldots, n$.  Here
$\ball{\ctrA}{\radA} \subseteq \Re^d$ denotes the (closed) ball with
center $\ctrA$ and radius $\radA \geq 0$.  Additionally, we are given
an approximation parameter $\eps \in (0,1)$.  For a point $\query \in
\Re^d$, the \emphi{distance} of $\query$ to a ball $\ballA =
\ball{\ctrA}{\radA}$ is
\begin{math}
    \dist{\query}{\ballA}%
    =%
    \max \pth{\MakeSBig \norm{\query - \ctrA} - \radA, \, 0 }.
\end{math}

\begin{observation}
    \obslab{ball:contained}%
    For two balls $\ballA_1 \subseteq \ballA_2 \subseteq \Re^d$, and
    any point $\query \in \Re^d$, we have $\dist{\query}{\ballA_1}
    \geq \dist{\query}{\ballA_2}$.
\end{observation}

The \emphi{$k$\th-nearest neighbor distance} of $\query$ to
$\BallSetA$, denoted by $\distPk{\BallSetA}{\query}{k}$, is the $k$\th
smallest number in $\dist{\query}{\ballA_1}, \dots,
\dist{\query}{\ballA_n}$. Similarly, for a given set of points
$\PntSet$, $\distPk{\PntSet}{\query}{k}$ denotes the $k$\th-nearest
neighbor distance of $\query$ to $\PntSet$.

We aim to build a data structure to answer $(1 \pm \eps)$-approximate
$k$\th-nearest neighbor (i.e., \emphi{$(k,\eps)$-\ANN}) queries, where
for any query point $\query \in \Re^d$ one needs to output a ball
$\ballA \in \BallSetA$ such that, $(1-\eps)
\distPk{\BallSetA}{\query}{k} \leq \dist{\query}{\ballA} \leq (1+\eps)
\distPk{\BallSetA}{\query}{k}$.  There are different variants
depending on whether $\eps$ and $k$ are provided with the query or in
advance.

We use \emphi{cube} to denote a set of the form
$[a_1, a_1 + \ell] \times [a_2, a_2 + \ell] \times \ldots \times [a_d, a_d + \ell] \subseteq \Re^d$, 
where $a_1, \ldots, a_d \in \Re$ and 
$\ell \geq 0$ is the side length of the cube.

\begin{observation}
    \obslab{lipschitz}
    For any set of balls $\BallSetA$, the function
    $\distPk{\BallSetA}{\query}{k}$ is a $1$-Lipschitz function; that
    is, for any two points $\pntA, \pntB$, we have that
    $\distPk{\BallSetA}{\pntA}{k} \leq \distPk{\BallSetA}{\pntB}{k} +
    \distE{\pntA}{\pntB}$.
\end{observation}

\begin{assumption}
    \assumplab{all:in:cube}%
    We assume all the balls are contained inside the cube
    $\pbrc[]{1/2-\delta, 1/2+\delta}^d$, which can be ensured by
    translation and scaling (which preserves order of distances),
    where $\delta = \eps/4$. As such, we can ignore queries outside
    the unit cube $[0,1]^d$, as any input ball is a valid answer in
    this case.
\end{assumption}

For a real positive number $\num$ and a point $\pnt = (\pnt_1, \ldots,
\pnt_d) \in \Re^d$, define $\Grid_\num(\pnt)$ to be the grid point
$\pth[]{\floor{\pnt_1/\num} \num, \ldots, \floor{\pnt_d/\num}
   \num}$. The number $\num$ is the \emphi{width} or
\emphi{side length} of the \emphi{grid} $\Grid_\num$. The mapping
$\Grid_\num$ partitions $\Re^d$ into cubes that are called grid
\emphic{cells}{cell}.

\begin{definition}
    \deflab{canonical:grid}%
    A cube is a \emphi{canonical cube} if it is contained inside the
    unit cube $\Scube=[0,1]^d$, it is a cell in a grid $\Grid_r$, and
    $r$ is a power of two (i.e., it might correspond to a node in a
    quadtree having $[0,1]^d$ as its root cell).  We will refer to
    such a grid $\Grid_r$ as a \emphic{canonical
       grid}{canonical!grid}. Note that all the cells corresponding
    to nodes of a compressed quadtree are canonical.
\end{definition}

\begin{definition}%
    \deflab{grid:approx}%
    Given a set $\ballA \subseteq \Re^d$, and a parameter $\delta >
    0$, let $\GridApproxX{\ballA}{\delta}$ denote the set of canonical
    grid cells of side length $2^{\floor{\log_2 {\delta
             \diamX[]{\ballA}/\sqrt{d}}}}$, that intersect $\ballA$,
    where $\diamX{\ballA} = \max_{\pnt,\pntA \in \ballA}
    \distE{\pnt}{\pntA}$ denotes the \emphi{diameter} of $\ballA$. Clearly,
    the diameter of any grid cell of $\GridApproxX{\ballA}{\delta}$,
    is at most $\delta \diamX[]{\ballA}$. Let
    $\GridSetB{\ballA} = \GridApproxX{\ballA}{1}$. It is easy to
    verify that $\cardin{\GridSetB{\ballA}} = O(1)$.  The set
    $\GridSetB{\ballA}$ is the \emphi{grid approximation} to
    $\ballA$.
\end{definition}

Let $\BallSetA$ be a family of balls in $\Re^d$. Given a set $\setA
\subseteq \Re^d$, let
\begin{align*}
    \capSubset{\BallSetA}{\setA} = \brc{ \ballA \in \BallSetA \sep{
          \ballA \cap \setA \neq \emptyset}}
\end{align*}
denote the set of all balls in $\BallSetA$ that intersect $\setA$.

For two compact sets $\setA, \setB \subseteq \Re^d$, $\setA \preceq
\setB$ if and only if $\diamX{\setA} \leq \diamX{\setB}$. For a set
$\setA$ and a set of balls $\BallSetA$, let
$\IntBallsX{\BallSetA}{\setA} = \brc{ \ballA \in \BallSetA \sep{
      \ballA \cap \setA \neq \emptyset \text{ and } \ballA \succeq
      \setA }}$.  Let $\cDim$ denote the maximum number of pairwise
disjoint balls of radius at least $\radA$, that may intersect a given
ball of radius $\radA$ in $\Re^d$. Clearly, we have
$\cardin{\IntBallsX{\BallSetA}{\ballA}} \leq \cDim$ for any ball
$\ballA$. We have the following bounds,
\begin{lemma}
  \lemlab{cdim:bound}
  $2 \leq \cDim \leq 3^d$ for all $d$.
\end{lemma}

\begin{proof}
  Let $\ballA = \ball{\ctrA}{\radA}$ be a given ball of radius $\radA$.
  For the lower bound we can take two balls both of radius $\radA$ which
  touch $\ballA$ at diametrically opposite points and lie outside
  $\ballA$. We now show the upper bound. Let $\BallSetA$ be a set
  of disjoint
  balls, each having radius at least $\radA$ and touching $\ballA$.
  Consider a ball $\ballA' \in \BallSetA$. If no point of the boundary
  of $\ballA'$ touches $\ballA$, then clearly $\ballA'$ contains
  $\ballA$ in its interior and it is easy to see that 
  $\cardin{\BallSetA} = 1$. As such we assume that all balls in
  $\BallSetA$ have some point of their boundary inside $\ballA$.
  Take any point $\pnt$ 
  of the boundary of $\ballA'$ such that $\pnt$ is in $\ballA$, 
  and consider a ball of radius $\radA$ that lies completely inside
  $\ballA'$, is of radius $\radA$ and is tangent to $\ballA'$ at
  $\pnt$. We can find such a ball for each ball in $\BallSetA$.
  Moreover, these balls are all disjoint.
  Thus we have $\cardin{\BallSetA}$ disjoint balls of radius exactly
  $\radA$ that touch $\ballA$. It is easy to see that all such balls
  are completely inside $\ball{\ctrA}{3\radA}$. By a simple volume
  packing bound it follows that $\cardin{\BallSetA} \leq 3^d$.
\end{proof}

\begin{definition}%
    \deflab{d:monotonic}%
    For a parameter $\delta \geq 0$, a function $f: \Re^+ \to \Re^+$
    is \emphi{ $\delta$-monotonic}, if for every $x \geq 0$,
    $f(x/(1+\delta)) \leq f(x)$.
\end{definition}

\section{Approximate range counting for balls}
\seclab{apprx:rcount}

\begin{datastructure}
    \dslab{everything}%
    For a given set of disjoint balls $\BallSetA = \brc{\ballA_1,\ldots,
       \ballA_n}$ in $\Re^d$, we build the following data structure,
    that is useful in performing several of the tasks at hand.
    \begin{compactenum}[\;(A)]
        \item \textbf{Store balls in a (compressed) quadtree}. %
        For $i = 1, 2, \dots, n$, let $\Env_i = \GridSetB{\ballA_i}$,
        and let $\Env = \bigcup_{i=1}^n \Env_i$ denote the union of
        these cells.  Let $\QTree$ be a compressed quadtree
        decomposition of $[0,1]^d$, such that all the cells of $\Env$
        are cells of $\QTree$.  We preprocess $\QTree$ to answer point
        location queries for the cells of $\Env$. 
        This takes $O( n \log n)$ time, see \cite{h-gaa-11}.
        
        \item \itemlab{big:balls}%
        \textbf{Compute list of ``large'' balls intersecting each
           cell.}  For each node $\nodeA$ of $\QTree$, there is a list
        of balls registered with it. Formally, \emphi{register} a ball
        $\ballA_i$ with all the cells of $\Env_i$. Clearly, each ball
        is registered with $O(1)$ cells, and it is easy to see that
        each cell has $O(1)$ balls registered with it, since the balls
        are disjoint.
        
        Next, for a cell $\cellA$ in $\QTree$ we compute a list
        storing $\IntBallsX{\BallSetA}{\cellA}$, and these balls are
        \emphi{associated} with this cell.  These lists are computed
        in a top-down manner. To this end, propagate from a node
        $\nodeA$ its list $\IntBallsX{\BallSetA}{\cellA}$ (which we
        assume is already computed) down to its children. For a node
        receiving such a list, it scans it, and keep only the balls
        that intersect its cell (adding to this list the balls already
        registered with this cell).  For a node $\node \in \QTree$,
        let $\BallsAssocX{\node}$ be this list.

        \item \textbf{Build compressed quadtree on centers of balls.}
        Let $\CenterSetA$ be the set of centers of the balls of
        $\BallSetA$.  Build, in $O(n \log n)$ time, a compressed
        quadtree $\QTreeCen$ storing $\CenterSetA$.

        \item \textbf{\ANN for centers of balls.}  %
        Build a data structure $\DS$, for answering $2$-approximate
        $k$-nearest neighbor distances on $\CenterSetA$, the set of
        centers of the balls, see \cite{hk-drhrp-12}, where $k$ and
        $\eps$ are provided with the query.  The data structure $\DS$,
        returns a point $\ctrA \in \CenterSetA$ such that,
        $\distPk{\CenterSetA}{\query}{k} \leq \dist{\query}{\ctrA}
        \leq 2\distPk{\CenterSetA}{\query}{k}$.

        \item \textbf{Answering approximate range searching for the
           centers of balls.}  %
        
        Given a query ball $\ballQ = \ball{\query}{\num}$ and a
        parameter $\delta > 0$, one can, using $\QTreeCen$, report
        (approximately), in $O( \log n + 1/\delta^d)$ time, the points
        in $\ballQ \cap \CenterSetA$. Specifically, the query process
        computes $O(1/\delta^d)$ sets of points, such that their union
        $\setA$, has the property that $ \ballQ \cap \CenterSetA
        \subseteq \setA \subseteq (1+\delta)\ballQ \cap \CenterSetA$,
        where $(1+\delta)\ballQ$ is the scaling of $\ballQ$ by a
        factor of $1+\delta$ around its center. Indeed, compute the
        set $\GridSetB{\ballQ}$, and then using cell queries in
        $\QTreeCen$ compute the corresponding cells (this takes $O(
        \log n)$ time). Now, descend to the relevant level of the
        quadtree to all the cells of the right size, 
        that intersect $\ballQ$. Clearly, the union of points 
        stored in their subtrees
        are the desired set.  This takes overall $O( \log n +
        1/\delta^d)$ time.
        
        A similar data structure for approximate range searching is
        provided by Arya and Mount \cite{am-ars-00}, and our
        description above is provided for the sake of completeness.
        
    \end{compactenum}
    Overall, it takes $O(n \log n)$ time to build this data structure.
\end{datastructure}
We denote the collection of data structures above 
by \dsref{everything} and where necessary,  
specific functionality it provides, say for finding the large balls 
intersecting a cell, by \dsref{everything} \itemref{big:balls}.

\subsection{Approximate range counting among balls}
\seclab{A:impl}

We need the ability to answer approximate range counting queries on a
set of disjoint balls. Specifically, given a set of disjoint 
balls $\BallSetA$, and a
query ball $\ballA$, the target is to compute the size of the set
$\ballA \cap \BallSetA = \brc{ \ballB \in \BallSetA \sep{ \ballB \cap
      \ballA \neq \emptyset }}$. To make this query computationally
fast, we allow an approximation.  More precisely, for a ball $\ballA$
a set $\ApproxX{\ballA}$ is a \emphi{$(1+\delta)$-ball} of $\ballA$, if
$\ballA \subseteq \ApproxX{\ballA} \subseteq (1+\delta)\ballA$, where
$(1+\delta)\ballA$ is the $(1+\delta)$-scaling of $\ballA$ around its
center.  The purpose here, given a query ball $\ballA$, is to compute
the size of the set $\ApproxX{\ballA} \cap \BallSetA$ for some
$(1+\delta)$-ball $\ApproxX{\ballA}$ of $\ballA$. 


\begin{lemma}%
    \lemlab{relevant:cells}%
    Given a compressed quadtree $\QTree$ of size $n$, a convex set
    $\setA$, and a parameter $\delta > 0$, one can compute the set of
    nodes in $\QTree$, that realizes $\GridApproxX{\setA}{\delta}$
    (see \defref{grid:approx}), in $O\pth{\log n + 1/\delta^d }$ time.
    Specifically, this outputs a set $\NodeSetA$ of nodes, of size
    $O\pth{1/\delta^d}$, such that their cells intersect
    $\GridApproxX{\setA}{\delta}$, and their parents cell diameter is
    larger than $\delta \diamX{\setA}$. Note that the cells in
    $\NodeSetA$ might be significantly larger if they are leaves of
    $\QTree$.
\end{lemma}

\begin{proof}
    Let $\GridAC = \GridApproxX{\setA}{1}$ be the grid approximation
    to $\setA$. Using cell queries on the compressed quadtree, one can
    compute the cells of $\QTree$ that corresponds to these canonical
    cells. Specifically, for each cube $\cellA \in \GridSetB{\setA}$,
    the query either returns a node for which this is its cell, or it
    returns a compressed edge of the quadtree; that is, two cells (one
    is a parent of the other), such that $\cellA$ is contained in of
    them and contains the other. Such a cell query takes $O( \log n)$
    time \cite{h-gaa-11}. This returns $O( 1)$ nodes in $\QTree$ such that
    their cells cover $\GridSetB{\setA}$.
    
    Now, traverse down the compressed quadtree starting from these
    nodes and collect all the nodes of the quadtree that are
    relevant. Clearly, one has to go at most $O( \log 1/\delta)$
    levels down the quadtree to get these nodes, and this takes
    $O(1/\delta^d)$ time overall.
\end{proof}

\begin{lemma}
    \lemlab{large:balls}%
    Let $\setA$ be any convex set in $\Re^d$, and let $\delta > 0 $ be
    a parameter. Using \dsref{everything}, one can compute, in
    $O\pth{\log n + 1/\delta^d}$ time, all the balls of $\BallSetA$
    that intersect $\setA$, with diameter $\geq \delta \diamX{
       \setA}$.
\end{lemma}

\begin{proof}
    We compute the cells of the quadtree realizing
    $\GridApproxX{\setA}{\delta}$ using \lemref{relevant:cells}. Now,
    from each such cell (and its parent), we extract the list of large
    balls intersecting it (there are $O(1/\delta^d)$ such nodes, and
    the size of each such list is $O(1)$). Next we check for each such
    ball if it intersects $\setA$ and if its diameter is at least
    $\delta \diamX{ \setA}$. We return the list of all such balls.
\end{proof}

\subsection{Answering a query}

Given a query ball $\ballQ = \ball{\query}{\num}$, and an
approximation parameter $\delta > 0$, our purpose is to compute a
number $\NxNum$, such that $\BSetNum{\MakeBig\query}{\num} \leq \NxNum
\leq \BSetNum{\MakeBig\query}{(1+\delta)\num}$.

The query algorithm works as follows:
\begin{compactenum}[\qquad(A)]
    \item Using \lemref{large:balls}, compute a set $\setA$ of all the
    balls that intersect $\ballQ$ and are of radius $\geq \delta
    \num/4$.
    
    \item Using \dsref{everything}, compute
    $O(1/\delta^d)$ cells of $\QTreeCen$ that corresponds to $
    \GridApproxX{\ballQ(1+\delta/4)}{\delta/4}$.  Let $\NxNum'$ be the
    total number of points in $\CenterSetA$ stored in these nodes.
    
    \item The quantity $\NxNum' + \cardin{\setA}$ is almost the
    desired quantity, except that we might be counting some of the
    balls of $\setA$ twice. To this end, let $\NxNum''$ be the number
    of balls in $\setA$ with centers in
    $\GridApproxX{\ballQ(1+\delta/4)}{\delta/4}$
    
    \item Let $\NxNum \leftarrow \NxNum' + \cardin{\setA} -
    \NxNum''$. Return $\NxNum$.
\end{compactenum}

We only sketch the proof, as the proof is straightforward.  Indeed,
the union of the cells of $\GridApproxX{\ballQ(1+\delta/4)}{\delta/4}$
contains $\ball{\query}{\num(1+\delta/4)}$ and is contained in
$\ball{\query}{(1+\delta)\num}$. All the balls with radius smaller
than $\delta \num/4$ and intersecting $\ball{\query}{\num}$ have their
centers in cells of $\GridApproxX{\ballQ(1+\delta/4)}{\delta/4}$, and
their number is computed correctly. Similarly, the ``large'' balls are
computed correctly. The last stage ensures we do not over-count by $1$
each large ball that also has its center in
$\GridApproxX{\ballQ(1+\delta/4)}{\delta/4}$. It is also easy to check
that $\BSetNum{\query}{\num} \leq
\NxNum \leq \BSetNum{\query}{\num(1+\delta)}$. The same result can be
used for $\num/(1+\delta)$ to get $\delta$-monotonicity of 
$\NxNum$. 

We now analyze the running time.
Computing all the cells of
$\GridApproxX{\ballQ(1+\delta/4)}{\delta/4}$ takes $O(\log n +
1/\delta^{d})$ time.  Computing the ``large'' balls takes $O\pth{\log
   n + 1/\delta^d}$ time. Checking for each large ball if it is
already counted by the ``small'' balls takes $O(1/\delta^d)$ by using
a grid. We denote the above query algorithm by 
\algBSetX{\query}{\delta}{\num}.

 The above implies the following.

\begin{lemma}
    \lemlab{Aimpl:main}%
    Given a set $\BallSetA$ of $n$ disjoint balls in $\Re^d$, it can be
    preprocessed, in $O(n \log n)$ time, into a data structure of size
    $O(n)$, such that given a query ball $\ball{\query}{\num}$ and
    approximation parameter $\delta > 0$, the query algorithm
    \algBSetX{\query}{\delta}{\num} returns, in $O(\log n +
    1/\delta^{d})$ time, a number $\NxNum$ satisfying the following:
    \begin{compactenum}[\qquad(A)]
        \item $\NxNum \leq \BSetNum{\query}{(1+\delta)\num}$,
        \item $\BSetNum{\query}{\num} \leq \NxNum$, and
        \item for a query ball $\ball{\query}{\cval}$ and $\delta$,
        the number $\NxNum$ is $\delta$-monotonic as a function of
        $\cval$, see \defref{d:monotonic}.
    \end{compactenum}
\end{lemma}


\section{Answering $k$-\ANN queries among balls}
\seclab{k:at:query} 

\subsection{Computing a constant factor approximation to
   $\distPk{\BallSetA}{\query}{k}$}

\begin{lemma}
    \lemlab{int:large}%
    Let $\BallSetA$ be a set of disjoint balls in $\Re^d$, and
    consider a ball $\ballA = \ball{\query}{r}$ that intersects at
    least $k$ balls of $\BallSetA$. Then, among the $k$ nearest
    neighbors of $\query$ from $\BallSetA$, there are at least
    $\max(0,k-\cDim)$ balls of radius at most $r$. The centers of all
    these balls are in $\ball{\query}{2r}$.
\end{lemma}
\begin{proof}
    Consider the $k$ nearest neighbors of $\query$ from
    $\BallSetA$. Any such ball that has its center outside
    $\ball{\query}{2r}$, has radius at least $r$, since it intersects
    $\ballA = \ball{\query}{r}$.  Since the number of balls that are
    of radius at least $r$ and intersecting $\ballA$ is bounded by
    $\cDim$, there must be at least $\max(0,k - \cDim)$ balls among
    the $k$ nearest neighbors, each having radius less than $r$.  Now,
    $\ball{\query}{2r}$ will contain the centers of all such balls.
\end{proof}

\begin{corollary}
    \corlab{int:large:cor}%
    Let $\numA = \min(k, \cDim)$. Then,
    $\distPk{\CenterSetA}{\query}{k - \numA}/2 \leq
    \distPk{\BallSetA}{\query}{k}$.
\end{corollary}


The basic observation is that we only need a rough approximation to
the right radius, as using approximate range counting (i.e.,
\lemref{Aimpl:main}), one can improve the approximation.

Let $\num_{i}$ denote the distance of $\query$ to the $i$\th closest
center in $\CenterSetA$. Let $\kdist = \distPk{\BallSetA}{\query}{k}$.
Let $i$ be the minimum index, such that $\kdist \leq x_i$.  Since
$\kdist \leq \num_k$, it must be that $i \leq k$.  There are several
possibilities:
\begin{compactenum}[\;\;(A)]
    \item If $i \leq k-\cDim$ (i.e., $\kdist \leq \num_{k - \cDim}$)
    then, by \lemref{int:large}, the ball $\ball{\query}{2 \kdist}$
    contains at least $k - \cDim$ centers. As such, $\kdist < \num_{k
       - \cDim} \leq 2 \kdist$, and $\num_{k - \cDim}$ is a good
    approximation to $\kdist$.

    \item If $i > k-\cDim$, and  $\kdist \leq 4 x_{i-1}$, then 
    $x_{i-1}$ is the desired approximation.

    \item If $i > k-\cDim$, and $\kdist \geq x_{i}/4$, then $x_{i}$ is
    the desired approximation.

    \item Otherwise, it must be that $i > k-\cDim$, and $4x_{i-1}
    < \kdist < x_{i}/4$.  Let $\ballA_j =
    \ball{\ctrA_j}{\radA_j}$ be the $j$\th closest ball to $\query$,
    for $j=1,\ldots, k$.  It must be that $\ballA_{i}, \ldots,
    \ballA_{k}$ are much larger than $\ball{\query}{\kdist}$.  But
    then, the balls $\ballA_{i}, \ldots, \ballA_{k}$ must intersect
    $\ball{\query}{x_{i}/2}$, and their radius is at least $x_i/2$.
    We can easily compute these big balls using \dsref{everything}
    \itemref{big:balls}, and the number of centers of the small balls
    close to query, and then compute $\kdist$ exactly.
\end{compactenum}

We build \dsref{everything} in $O(n \log n)$ time.

First we introduce some notation.
For $\num \geq 0$, let $\NxNumX{\num}$ denote the number of balls in
$\BallSetA$ that intersect $\ball{\query}{\num}$; that is
$\NxNum(\num) = \cardin{\brc{\ballA \in \BallSetA \sep{ \ballA \cap
         \ball{\query}{\num} \neq \emptyset}}}$, and $\CxNum(\num)$
denote the number of centers in $\ball{\query}{\num}$, i.e.,
$\CxNum(\num) = \cardin{\CenterSetA \cap \ball{\query}{\num}}$.  Also,
let $\bcountX{\num}$ denote the $2$-approximation to the number of
balls of $\BallSetA$ intersecting $\ball{\query}{\num}$, as computed
by \lemref{Aimpl:main}; that is $\NxNumX{\num} \leq \bcountX{\num}
\leq \NxNumX{2\num}$.

We now provide our algorithm to answer a query.
We are given a query point $\query \in \Re^d$ and a number $k$.

Using \dsref{everything}, compute a $2$-approximation for the smallest
ball containing $k-i$ centers of $\BallSetA$, for $i=0,\ldots, \numA$,
where $\numA = \min(k,\cDim)$, and let $r_{k-i}$ be this radius. That
is, for $i=0,\ldots, \numA$, we have $\CxNum(r_{k-i}/2) \leq k -i \leq
\CxNum(r_{k-i})$.  For $i=0, \ldots, \numA$, compute $N_{k-i} =
\bcountX{r_{k-i}}$ (\lemref{Aimpl:main}).

Let $\alpha$ be the maximum index such that $N_{k-\alpha} \geq k$.
Clearly, $\alpha$ is well defined as $N_k \geq k$. The algorithm is
executed in the following steps.
\begin{compactenum}[\quad(A)]
    \item \itemlab{step:A} %
    If $\alpha = \numA$ we return $2 r_{k - \numA}$.
    \item \itemlab{step:B} If $\bcountX{r_{k-\alpha}/4} < k$, we
    return $2 r_{k-\alpha}$.
    \item \itemlab{step:C} Otherwise, compute all the balls of
    $\BallSetA$ that are of radius at least $r_{k - \alpha}/4$ and
    intersect the ball $\ball{\query}{r_{k-\alpha}/4}$, using
    \dsref{everything} \itemref{big:balls}.  For each such ball
    $\ballA$, compute the distance $\numB = \dist{\query}{\ballA}$ of
    $\query$ to it.  Return $2 \numB$ for the minimum such number
    $\numB$ such that $\bcountX{\numB} \geq k$.
\end{compactenum}

\begin{lemma}
    \lemlab{const:fact:apprx}
    Given a set of $n$ disjoint balls $\BallSetA$ in $\Re^d$, one can
    preprocess them, in $O(n \log n)$ time, into a data structure of
    size $O(n)$, such that given a query point $\query \in \Re^d$, and
    a number $k$, one can compute, in $O(\log n)$ time, a number
    $\num$ such that, $\num/4 \leq \distPk{\BallSetA}{\query}{k} \leq
    4 \num$.
\end{lemma}

\begin{proof}
    The data structure and query algorithm are described above. We
    next prove correctness.  To prove that \itemref{step:A} returns
    the correct answer observe that under the given assumptions,
    \[
    r_{k - \numA}/4 \leq \distPk{\CenterSetA}{\query}{k - \numA} / 2
    \leq \distPk{\BallSetA}{\query}{k} \leq 2 r_{k - \numA},
    \]
    where the second inequality follows from \corref{int:large:cor},
    and the third inequality follows as $\NxNum(2 r_{k - \numA}) \geq
    \bcountX{ r_{k - \numA} } \geq k$, while
    $\distPk{\BallSetA}{\query}{k}$ is the smallest number $\num$ such
    that $\NxNum(\num) \geq k$.
    
    For \itemref{step:B} observe that we have that $\NxNum(r_{k -
       \numA}/4) \leq \bcountX{r_{k - \numA}/4} < k$ and as such we
    have $r_{k - \numA}/4 < \distPk{\BallSetA}{\query}{k}$. But by
    assumption, $\bcountX{r_{k - \numA}} \geq k$ and so $\NxNum(2 r_{k
       - \numA}) \geq \bcountX{r_{k - \numA}} \geq k$, thus
    $\distPk{\BallSetA}{\query}{k} \leq 2 r_{k - \numA}$.
    
    For \itemref{step:C}, first observe that $\alpha < \numA$ as the
    algorithm did not return in \itemref{step:A}.  Since $\alpha$ is
    the maximum index such that $\bcountX{r_{k - \alpha}} \geq k$, so
    $\NxNum(r_{k - \alpha - 1}) \leq \bcountX{r_{k - \alpha - 1}} < k$
    implying, $r_{k - \alpha - 1} < \distPk{\BallSetA}{\query}{k}$.
    Also, $\distPk{\BallSetA}{\query}{k} \leq r_{k - \alpha}/4$, as
    the algorithm did not return in \itemref{step:B}. Now the ball
    $\ball{\query}{r_{k - \alpha - 1}}$ contains at least $k - \alpha
    - 1$ centers from $\CenterSetA$, but it does not contain $k -
    \alpha$ centers. Indeed, otherwise we would have
    $\distPk{\CenterSetA}{\query}{k - \alpha} \leq r_{k - \alpha - 1}$
    and so $r_{k - \alpha} \leq 2 \distPk{\CenterSetA}{\query}{k -
       \alpha} \leq 2 r_{k - \alpha - 1}$, but on the other hand $r_{k
       - \alpha - 1} < \distPk{\BallSetA}{\query}{k} \leq r_{k -
       \alpha}/4$, which would be a contradiction. Similarly, there is
    no center of any ball whose distance from $\query$ is in the range
    $(r_{k - \alpha - 1}, r_{k - \alpha}/2)$ otherwise we would have
    that $\distPk{\CenterSetA}{\query}{k-\alpha} < r_{k - \alpha}/2$
    and this would mean that $r_{k - \alpha} \leq 2
    \distPk{\CenterSetA}{\query}{k-\alpha} < r_{k - \alpha}$, a
    contradiction.  Now, the center of the $k$\th closest ball is
    clearly more than $r_{k - \alpha - 1}$ away from $\query$. As such
    its distance from $\query$ is at least $r_{k - \alpha}/2$.  Since
    $\distPk{\BallSetA}{\query}{k} \leq r_{k - \alpha}/4$ it follows
    that the $k$\th closest ball intersects $\ball{\query}{r_{k -
          \alpha}/4}$ and moreover, its radius is at least $r_{k -
       \alpha}/4$. Since we compute all such balls in
    \itemref{step:C}, we do encounter the $k$\th closest ball. It is
    easy to see that in this case we return a number $\numB$
    satisfying, $\numB/2 \leq \distPk{\BallSetA}{\query}{k} \leq 2
    \numB$.

    \smallskip

    As for the running time, notice that we need to use the algorithm
    of \lemref{Aimpl:main} $O(1)$ times, each iteration taking time
    $O(\log n)$. After this we need another $O(\log n)$ time for the
    invocation of the algorithm in \lemref{large:balls}. As such, the
    total query time is $O(\log n)$.
\end{proof}

We now show how to refine the approximation.

\begin{lemma}
    \lemlab{count:range}%
    Given a set $\BallSetA$ of $n$ balls in $\Re^d$, it can be
    preprocessed, in $O(n \log n)$ time, into a data structure of size
    $O(n)$. Given a query point $\query$, numbers $k, \num$, and an
    approximation parameter $\eps >0$, such that $\num/4 \leq
    \distPk{\BallSetA}{\query}{k} \leq 4 \num$, one can find a ball
    $\ballA \in \BallSetA$ such that, $(1-\eps)
    \distPk{\BallSetA}{\query}{k} \leq \dist{\query}{\ballA} \leq
    (1+\eps) \distPk{\BallSetA}{\query}{k}$, in $O\pth{ \log n
       +1/\eps^d}$ time.
\end{lemma}

\begin{proof}
    We are going to use the same data structure as
    \lemref{Aimpl:main}, for the query ball $\ballQ =
    \ball{\query}{4\num(1+\eps)}$.  We compute all large balls of
    $\BallSetA$ that intersect $\ballQ$.  Here a large ball is a ball
    of radius $> \num \eps$, and a ball of radius at most $\num \eps$
    is considered to be a small ball. Consider the $O(1/\eps^d)$ grid
    cells of $\GridApproxX{\ballQ}{\eps/16}$.  In $O(1/\eps^d)$ time
    we can record the number of centers of large balls inside any such
    cell.  Clearly, any small ball that intersects $\ball{\query}{4
       \num}$ has its center in some cell of
    $\GridApproxX{\ballQ}{\eps/16}$. We use the quadtree $\QTreeCen$
    to find out exactly the number of centers, $N_\cellA$, of small
    balls in each cell $\cellA$ of $\GridApproxX{\ballQ}{\eps/16}$, by
    finding the total number of centers using $\QTreeCen$, and
    decreasing this by the count of centers of large balls in that
    cell. This can be done in time $O(\log n + 1/\eps^d)$.  We pick an
    arbitrary point in $\cellA$, and assign it weight $N_\cellA$, and
    treat it as representing all the small balls in this grid cell --
    clearly, this introduces an error of size $\leq \eps \num$ in the
    distance of such a ball from $\query$, and as such we can ignore
    it in our argument. In the end of this snapping process, we have
    $O(1/\eps^d)$ weighted points, and $O(1/\eps^d)$ large balls. We
    know the distance of the query point from each one of these
    points/balls. This results in $O(1/\eps^d)$ weighted distances,
    and we want the smallest $\ell$, such that the total weight of the
    distances $\leq \ell$ is at least $k$. This can be done by
    weighted median selection in linear time in the number of
    distances, which is $O(1/\eps^d)$. Once we get the required point
    we can output any ball $\ballA$ corresponding to the
    point. Clearly, $\ballA$ satisfies the required conditions.
\end{proof}

\subsection{The result}

\begin{theorem}
    \thmlab{k:at:query:main}%
    Given a set of $n$ disjoint balls $\BallSetA$ in $\Re^d$, one can
    preprocess them in time $O(n \log n)$ into a data structure of
    size $O(n)$, such that given a query point $\query \in \Re^d$, a
    number $k$ with $1 \leq k \leq n$ and $\eps > 0$, one can find in
    time $O\pth{ \log n + \eps^{-d}}$ a ball $\ballA \in \BallSetA$,
    such that, $(1-\eps)\distPk{\BallSetA}{\query}{k} \leq
    \dist{\query}{\ballA} \leq (1+\eps)
    \distPk{\BallSetA}{\query}{k}$.
\end{theorem}

\section{Quorum clustering}
\seclab{quorum}

We are given a set $\BallSetA$ of $n$ disjoint balls in $\Re^d$, and
we describe how to compute quorum clustering for them quickly.
\smallskip

Let $\constC$ be some constant.  Let $\BallSetA_0 = \emptyset$. For $i
= 1,\ldots, \nclust$, let $\BallSetX_i = \BallSetA \setminus
(\bigcup_{j=0}^{i-1} \BallSetA_j)$, and let $\ballC_i =
\ball{\ctrC_i}{\radC_i}$ be any ball that satisfies,
\begin{compactenum}[\qquad(A)]
    \item $\ballC_i$ contains $\min(k - \cDim, \cardin{\BallSetX_i})$
    balls of $\BallSetX_i$ completely inside it,
    \item $\ballC_i$ intersects at least $k$ balls of $\BallSetA$, and
    \item the radius of $\ballC_i$ is at most $\constC$ times the
    radius of the smallest ball satisfying the above conditions.
\end{compactenum}
Next, we remove any $k - \cDim$ balls that are contained in $\ballC_i$
from $\BallSetX_i$ to get the set $\BallSetX_{i+1}$. We call the removed
set of balls $\BallSetA_i$. We repeat this
process till all balls are extracted.  Notice that at each step $i$,
we only require that the $\ballC_i$ intersects $k$ balls of
$\BallSetA$ (and not $\BallSetX_i$), but that it must contain $k -
\cDim$ balls from $\BallSetX_i$.  Also, the last quorum ball may
contain fewer balls. The balls $\ballC_1, \ldots, \ballC_{\nclust}$,
are the resulting \emphi{$\constC$-approximate quorum clustering}.

\subsection{Computing an approximate quorum clustering}
\begin{definition}
    For a set $\PntSet$ of $n$ points in $\Re^d$, and an integer
    $\ell$, with $1 \leq \ell \leq n$, let $\minDistPk{\PntSet}{\ell}$
    denote the radius of the smallest ball which contains at least
    $\ell$ points from $\PntSet$, i.e., $\minDistPk{\PntSet}{\ell} =
    \min_{\query \in \Re^d} \distPk{\PntSet}{\query}{\ell}$.

    Similarly, for a set $\BallSetX$ of $n$ balls in $\Re^d$, and an
    integer $\ell$, with $1 \leq \ell \leq n$, let
    $\minBCoverk{\BallSetX}{\ell}$ denote the radius of the smallest
    ball which completely contains at least $\ell$ balls from
    $\BallSetX$.
\end{definition}

\begin{lemma}[\cite{hk-drhrp-12}]
    \lemlab{q:clustering:ref}
    Given a set $\PntSet$ of $n$ points in $\Re^d$ and integer $\ell$,
    with $1 \leq \ell \leq n$, one can compute, in $O(n \log n)$ time,
    a sequence of $\ceiling{n/\ell}$ balls, $\ballD_1 =
    \ball{\ctrD_1}{\radD_1}, \ldots, \ballD_{\ceiling{n/\ell}} =
    \ball{\ctrD_{\ceiling{n/\ell}}}{\radD_{\ceiling{n/\ell}}}$, such
    that, for all $i, 1 \leq i \leq \ceiling{n/\ell}$, we have
    \smallskip
    \begin{compactenum}[\rm \quad(A)]
        \item For every ball $\ballD_i$, there is an associated subset
        $\PntSet_i$ of $\min(\ell, \cardin{\PntSetQ_i})$ points of
        $\PntSetQ_i = \PntSet\setminus \pth[]{ \PntSet_i \cup \ldots
           \cup \PntSet_{i-1}}$, that it covers.
        \item The ball $\ballD_i = \ball{\ctrD_i}{\radD_i}$ is a
        $2$-approximation to the smallest ball covering
        $\min(\ell,\cardin{\PntSetQ_i})$ points in $\PntSetQ_i$; that
        is, $\radD_i/2 \leq
        \minDistPk{\PntSetQ_i}{\min(\ell,\cardin{\PntSetQ_i})} \leq
        \radD_i$.
    \end{compactenum}
\end{lemma}


The algorithm to construct an approximate quorum clustering is as
follows.  We use the algorithm of \lemref{q:clustering:ref} with the
set of points $\PntSet = \CenterSetA$, and $\ell = k - \cDim$ to get a
list of $\nclust = \ceiling{n/(k - \cDim)}$ balls $\ballD_1 =
\ball{\ctrD_1}{\radD_1}, \ldots, \ballD_\nclust =
\ball{\ctrD_\nclust}{\radD_\nclust}$, satisfying the conditions of
\lemref{q:clustering:ref}. Next we use the algorithm of
\thmref{k:at:query:main}, to compute $(k,\eps)$-ANN distances from the
centers $\ctrD_1, \ldots, \ctrD_\nclust$, to the balls of $\BallSetA$.

Thus, we get numbers $\numA_i$ satisfying, $(1/2)
\distPk{\BallSetA}{\ctrD_i}{k} \leq \numA_i \leq (3/2)
\distPk{\BallSetA}{\ctrD_i}{k}$.  Let $\numB_i = \max(2 \numA_i, 3
\radD_i)$, for $i=1,\ldots,\nclust$.  Sort $\numB_1, \ldots,
\numB_{\nclust}$ (we assume for the sake of simplicity of exposition
that $\numB_{\nclust}$, being the radius of the last cluster is the
largest number).  Suppose the sorted order is the permutation $\pi$ of
$\brc{1,\ldots, \nclust}$ (by assumption $\pi(\nclust) = \nclust$).
We output the balls $\ballC_i =
\ball{\ctrD_{\pi(i)}}{\numB_{\pi(i)}}$, for $i = 1,\ldots, \nclust$,
as the approximate quorum clustering.

\subsection{Correctness}
\begin{lemma}
    \lemlab{times:three}
    Let $\BallSetA = \brc{\ballA_1, \ldots, \ballA_n}$ be a set of $n$
    disjoint balls, where $\ballA_i = \ball{\ctrA_i}{\radA_i}$, for $i
    = 1, \ldots, n$.  Let $\CenterSetA = \brc{\ctrA_1,\ldots,
       \ctrA_n}$ be the set of centers of these balls.  Let $\ballA =
    \ball{\ctrA}{\radA}$ be any ball that contains at least $\ell$
    centers from $\CenterSetA$, for some $2 \leq \ell \leq n$.  Then
    $\ball{\ctrA}{3 \radA}$ contains the $\ell$ balls that correspond
    to those centers.
\end{lemma}
\begin{proof}
    Without loss of generality suppose $\ballA$ contains the $\ell$
    centers $\ctrA_1, \ldots, \ctrA_\ell$, from $\CenterSetA$. Now
    consider any index $i$ with $1 \leq i \leq \ell$, and consider any
    $j \neq i$, which exists as $\ell \geq 2$ by assumption. Since
    $\ball{\ctrA}{\radA}$ contains both $\ctrA_i$ and $\ctrA_j$,
    $2 \radA \geq \norm{\ctrA_i - \ctrA_j}$ by the triangle
    inequality.  On the other hand, as the balls $\ballA_i$ and
    $\ballA_j$ are disjoint we have that $\norm{\ctrA_i - \ctrA_j}
    \geq \radA_i + \radA_j \geq \radA_i$.  It follows that $\radA_i
    \leq 2 \radA$ for all $1 \leq i \leq \ell$.  As such the ball
    $\ball{\ctrA}{3 \radA}$ must contain the entire ball $\ballA_i$,
    and thus it contains all the $\ell$ balls $\ballA_1, \ldots,
    \ballA_\ell$, corresponding to the centers.
\end{proof}

\begin{lemma}
    \lemlab{balls:packing}%
    Let $\BallSetA = \brc{\ballA_1 = \ball{\ctrA_1}{\radA_1}, \ldots,
       \ballA_n = \ball{\ctrA_n}{\radA_n}}$ be a set of $n$ disjoint
    balls in $\Re^d$.  Let $\CenterSetA = \brc{\ctrA_1,\ldots,
       \ctrA_n}$ be the corresponding set of centers, and let $\ell$
    be an integer with $2 \leq \ell \leq n$. Then,
    $\minDistPk{\CenterSetA}{\ell} \leq \minBCoverk{\BallSetA}{\ell}
    \leq 3 \minDistPk{\CenterSetA}{\ell} $.
\end{lemma}

\begin{proof}
    The first inequality follows since the ball realizing the optimal
    covering of $\ell$ balls, clearly contains their centers as well,
    and therefore $\ell$ points from $\CenterSetA$.  To see the second
    inequality, consider the ball $\ballA = \ball{\ctrA}{\radA}$
    realizing $\minDistPk{\CenterSetA}{\ell}$, and use
    \lemref{times:three} on it. This implies
    $\minBCoverk{\BallSetA}{\ell} \leq 3
    \minDistPk{\CenterSetA}{\ell}$.
\end{proof}

\begin{lemma}
    \lemlab{q:clustering:corr}
    The balls $\ballC_1, \ldots \ballC_\nclust$ computed above are a
    $12$-approximate quorum clustering of $\BallSetA$.
\end{lemma}

\begin{proof}
    Consider the balls $\ballD_1 = \ball{\ctrD_1}{\radD_1}, \ldots,
    \ballD_\nclust=\ball{\ctrD_\nclust}{\radD_\nclust}$ computed by
    the algorithm of \lemref{q:clustering:ref}. Suppose
    $\CenterSetA_i$, for $1 = 1,\ldots, \nclust$, is the set of
    centers assigned to the balls $\ballA_i$. That is $\CenterSetA_1,
    \ldots, \CenterSetA_\nclust$ form a disjoint decomposition of
    $\CenterSetA$, each of size $k - \cDim$ (except for the last set,
    which might be smaller -- a technicality that we ignore for the
    sake of simplicity of exposition).

    For $i = 1, \ldots, \nclust$, let $\BallSetA_i$ denote the set of
    balls corresponding to the centers in $\CenterSetA_i$.  Now while
    constructing the approximate quorum clusters we are going to
    assign the set of balls $\BallSetA_{\pi(i)}$ for
    $i=1,\ldots,\nclust$, to $\ballC_i$. Now, fix a $i$ with $1 \leq i
    \leq \nclust - 1$. The balls of $\bigcup_{j=1}^{i}
    \BallSetA_{\pi(j)}$ have been used up. Consider an optimal ball,
    i.e., a ball $\ballA = \ball{\ctrA}{\radA}$ that contains
    completely $k - \cDim$ balls among $\bigcup_{j=i+1}^{\nclust}
    \BallSetA_{\pi(j)}$ and intersects $k$ balls from $\BallSetA$, and
    is the smallest such possible.  Fix some $k - \cDim$ balls from
    $\bigcup_{j=i+1}^{\nclust} \BallSetA_{\pi(j)}$ that this optimal
    ball contains. Consider the sets of centers $\CenterSetA'$ of
    these balls. The quorum clusters $\ballD_{\pi(j)}$ for $j =
    i+1,\ldots,\nclust$, contain all these centers, by
    construction. Out of these indices, i.e., out of the indices
    $\brc{\pi(i+1),\ldots,\pi(\nclust)}$, suppose $p$ is the minimum
    index such that $\ballD_p$ contains one of these centers.  When
    $\ballD_p$ was constructed, i.e., at the $p$\th iteration of the
    algorithm of \lemref{q:clustering:ref}, all the centers from
    $\CenterSetA'$ were available. Now since the optimal ball
    $\ballA=\ball{\ctrA}{\radA}$ contains $k - \cDim$ available
    centers too, it follows that $\radD_p \leq 2 \radA$ since
    \lemref{q:clustering:ref} guarantees this. Since $k - \cDim \geq
    2$, by \lemref{times:three}, $\ball{\ctrD_p}{3 \radD_p}$ contains
    the balls of $\BallSetA_p$. Moreover, by the Lipschitz property,
    see \obsref{lipschitz}, it follows that
    $\distPk{\BallSetA}{\ctrD_p}{k} \leq \distPk{\BallSetA}{\ctrA}{k}
    + \norm{\ctrD_p - \ctrA} \leq \radA + (\radA + \radD_p) \leq 4
    \radA$, where the second last inequality follows as the balls
    $\ballA=\ball{\ctrA}{\radA}$ and the ball $\ballD_p =
    \ball{\ctrD_p}{\radD_p}$ intersect. Therefore, for the index $p$
    we have that, $\distPk{\BallSetA}{\ctrD_p}{k} \leq 2 \numA_p \leq
    3 \distPk{\BallSetA}{\ctrD_p}{k} \leq 12 \radA$, and also that $3
    \radD_p \leq 6 \radA$. As such $\numB_p = \max(2 \numA_p, 3
    \radD_p) \leq 12 \radA$. The index $\pi(i+1)$ minimizes this
    quantity among the indices $\brc{\pi(i+1),\ldots,\pi(m)}$ (as we
    took the sorted order), as such it follows that $\numB_{i+1} \leq
    12 \radA$.
\end{proof}

\begin{lemma}
    \lemlab{q:clustering}
    Given a set $\BallSetA$ of $n$ disjoint balls in $\Re^d$, such
    that $(k-\cDim) | n$, and a number $k$ with $2 \cDim < k \leq n$,
    in $O(n \log n)$ time, one can output a sequence of $\nclust = {n
       / (k - \cDim)}$ balls $\ballC_1, \ldots, \ballC_\nclust$, such
    that
    \begin{compactenum}[\quad(A)]
        \item For each ball $\ballC_i$, there is an associated subset
        $\BallSetA_i$ of $k - \cDim$ balls of $\BallSetX_i = \BallSetA
        \setminus (\BallSetA_1 \cup \ldots \cup \BallSetA_{i-1})$,
        that it completely covers.
        \item The ball $\ballC_i$ intersects at least $k$ balls from
        $\BallSetA$.

        \item The radius of the ball $\ballC_i$ is at most $12$ times
        that of the smallest ball covering $k - \cDim$ balls of
        ${\BallSetX_i}$ completely, and intersecting $k$ balls of
        $\BallSetA$.
    \end{compactenum}
\end{lemma}

\begin{proof}
    The correctness was proved in \lemref{q:clustering:corr}.  To see
    the time bound is also easy as the computation time is dominated
    by the time in \lemref{q:clustering:ref}, which is $O(n \log n)$.
\end{proof}

\section{Construction of the sublinear space data structure for
   $(k,\eps)$-\ANN}
\seclab{avd}

Here we show how to compute an approximate Voronoi diagram for
approximating the $k$\th-nearest ball, that takes $O(n/k)$ space. We
assume $k > 2 \cDim$ without loss of generality, and we let $\nclust =
\ceiling{n / (k - \cDim)} = O(n/k)$. Here $k$ and $\eps$ are
prespecified in advance.

\subsection{Preliminaries}

The following notation was introduced in \cite{hk-drhrp-12}.  A ball
$\ballA$ of radius $\radA$ in $\Re^d$, centered at a point $\ctrA$,
can be interpreted as a point in $\Re^{d+1}$, denoted by
$\mapped{\ballA} = \pth[]{ \ctrA, \radA}$. For a regular point $\pnt
\in \Re^d$, its corresponding image under this transformation is the
\emphi{mapped} point $\mapped{\pnt} = \pth[]{\pnt, 0 } \in \Re^{d+1}$,
i.e., we view it as a ball of radius $0$ and use the mapping defined
on balls.  Given point $\pntA = \pth{\pntA_1,\dots,\pntA_d} \in \Re^d$
we will denote its Euclidean norm by $\norm{\pntA}$.  We will consider
a point $\pntA = \pth{\pntA_1, \pntA_2,\dots, \pntA_{d+1}} \in
\Re^{d+1}$ to be in the product metric of $\Re^d \times \Re$ and
endowed with the product metric norm
\begin{align*}
    \normP{\pntA} = \sqrt{\pntA_1^2 + \dots + \pntA_d^2} +
    \abs{\pntA_{d+1}}.
\end{align*}
It can be verified that the above defines a norm, and for any $\pntA
\in \Re^{d+1}$ we have $\norm{\pntA} \leq \normP{\pntA} \leq \sqrt{2}
\norm{\pntA}$.

\subsection{Construction}

The input is a set $\BallSetA$ of $n$ disjoint balls in $\Re^d$, and
parameters $k$ and $\eps$.

The construction of the data structure is similar to the construction
of the $k$\th-nearest neighbor data structure from the authors' paper
\cite{hk-drhrp-12}.  We compute, using \lemref{q:clustering}, a
$\constC$-approximate quorum clustering of $\BallSetA$ with $\nclust =
{n / (k - \cDim)} = O(n/k)$ balls, $\QClustBalls=\brc{\ballC_1 =
   \ball{\ctrC_1}{\radC_1}, \ldots, \ballC_\nclust =
   \ball{\ctrC_\nclust}{\radC_\nclust}}$, where $\constC \leq 12$.
The algorithm then continues as follows:
\begin{compactenum}[(A)]
    \item Compute an exponential grid around each quorum
    cluster. Specifically, let
    \begin{align}
        \ds \CellSetA =\, \bigcup_{i = 1}^{\nclust} \;\;\bigcup_{j
           =0}^{ \ceiling{\log\pth[]{32\constC/\eps}}}
        \GridApproxX{\ball{\ctrC_i}{2^j \radC_i} }{
           \frac{\eps}{\constA}}%
        \eqlab{clusters:around:q}
    \end{align}
    be the set of grid cells covering the quorum clusters and their
    immediate environ, where $\constA$ is a sufficiently large
    constant (say, $\constA = 256 \constC$).
    
    \item Intuitively, $\CellSetA$ takes care of the region of space
    immediately next to a quorum cluster%
    \footnote{That is, intuitively, if the query point falls into one
       of the grid cells of $\CellSetA$, we can answer a query in
       constant time.}.  For the other regions of space, we can apply
    a construction of an approximate Voronoi diagram for the centers
    of the clusters (the details are somewhat more involved). To this
    end, lift the quorum clusters into points in $\Re^{d+1}$, as
    follows
    \begin{align*}
        \mapped{\QClustBalls} = \brc{\mapped{\ballC_1}, \dots,
           \mapped{\ballC_{\nclust}}},
    \end{align*}
    where $\mapped{\ballC_i} = \pth{\ctrC_i,\radC_i} \in \Re^{d+1}$,
    for $i=1,\ldots, \nclust$.  Note that all points in
    $\mapped{\QClustBalls}$ belong to $\mapped{\Scube} = [0,1]^{d+1}$
    by \assumpref{all:in:cube}.  Now build a $(1+\eps/8)$-\AVD for
    $\mapped{\QClustBalls}$ using the algorithm of Arya and Malamatos
    \cite{am-lsavd-02}, for distances specified by the $\normP{\cdot}$
    norm. The \AVD construction provides a list of canonical cubes
    covering $[0,1]^{d+1}$ such that in the smallest cube containing
    the query point, the associated point of $\mapped{\QClustBalls}$,
    is a $(1+\eps/8)$-\ANN to the query point. (Note that these cubes
    are not necessarily disjoint. In particular, the smallest cube
    containing the query point $\query$ is the one that determines the
    assigned approximate nearest neighbor to $\query$.)

    Clip this collection of cubes to the hyperplane $x_{d+1} = 0$
    (i.e., throw away cubes that do not have a face on this
    hyperplane). For a cube $\cellA$ in this collection, denote by
    $\repX{\cellA}$, the point of $\mapped{\QClustBalls}$ assigned to
    it.  Let $\CellSetC$ be this resulting set of canonical
    $d$-dimensional cubes.

    \item Let $\CellSetAVD$ be the space decomposition resulting from
    overlaying the two collection of cubes, i.e. $\CellSetA$ and
    $\CellSetC$.  Formally, we compute a compressed quadtree $\QTree$
    that has all the canonical cubes of $\CellSetA$ and $\CellSetC$ as
    nodes, and $\CellSetAVD$ is the resulting decomposition of space
    into cells. One can overlay two compressed quadtrees representing
    the two sets in linear time \cite{bhst-sqgqi-10, h-gaa-11}.  Here,
    a cell associated with a leaf is a canonical cube, and a cell
    associated with a compressed node is the set difference of two
    canonical cubes. Each node in this compressed quadtree contains
    two pointers -- to the smallest cube of $\CellSetA$, and to the
    smallest cube of $\CellSetC$, that contains it. This information
    can be computed by doing a \BFS on the tree.
    
    For each cell $\cellA \in \CellSetAVD$ we store the following.
    \begin{compactenum}[\qquad(I)]
        \item An arbitrary representative point $\repres{\cellA} \in
        \cellA$.
        
        \item The point $\repX{\cellA} \in \mapped{\QClustBalls}$ that
        is associated with the smallest cell of $\CellSetC$ that
        contains this cell. We also store an arbitrary ball,
        $\ballrepX{\cellA} \in \BallSetA$, that is one of the balls
        completely inside the cluster specified by $\repX{\cellA}$ --
        we assume we stored such a ball inside each quorum cluster,
        when it was computed.
        
        \item A number $\adknn{\repres{\cellA}}$ that satisfies
        $\distPk{\BallSetA}{\repres{\cellA}}{k} \leq
        \adknn{\repres{\cellA}} \leq
        (1+\eps/4)\distPk{\BallSetA}{\repres{\cellA}}{k}$, and a ball
        $\kbannrepX{\repres{\cellA}} \in \BallSetA$ that realizes this
        distance. In order to compute $\adknn{\repres{\cellA}}$ and
        $\kbannrepX{\repres{\cellA}}$ use the data structure of
        \secref{k:at:query}, see \thmref{k:at:query:main}.
    \end{compactenum}
\end{compactenum}

\subsection{Answering a query}

Given a query point $\query$, compute the leaf cell (equivalently the
smallest cell) in $\CellSetAVD$ that contains $\query$ by performing a
point-location query in the compressed quadtree $\QTree$.  Let
$\cellA$ be this cell.  Let,
\begin{align}
    \estX = \min \pth{\normP{\mapped{\query} - \repX{\cellA}},
       \adknn{\repres{\cellA}} + \norm{\query - \repres{\cellA}} }.
    \eqlab{in:cell}
\end{align}
If $\diameter{\cellA} \leq (\eps/8) \estX$ we return
$\kbannrepX{\repres{\cellA}}$ as the approximate $k$\th-nearest
neighbor, else we return $\ballrepX{\cellA}$.
\subsection{Correctness}

\begin{lemma}
    \lemlab{triv:ub} The number $\estX = \min
    \pth{\normP{\mapped{\query} - \repX{\cellA}},
       \adknn{\repres{\cellA}} + \norm{\query - \repres{\cellA}} }$
    satisfies, $\distPk{\BallSetA}{\query}{k} \leq \estX$.
\end{lemma}

\begin{proof}
    This follows by the Lipschitz property, see \obsref{lipschitz}.
\end{proof}

\begin{lemma}
    \lemlab{small:cell}%
    Let $\cellA \in \CellSetAVD$ be any cell containing $\query$. If
    $\diameter{\cellA} \leq \eps \distPk{\BallSetA}{\query}{k}/4$,
    then $\kbannrepX{\repres{\cellA}}$ is a valid $(1 \pm
    \eps)$-approximate $k$\th-nearest neighbor of $\query$.
\end{lemma}

\begin{proof}
    For the point $\repres{\cellA}$, by \obsref{lipschitz}, we have
    that
    \begin{align*}
        \distPk{\BallSetA}{\repres{\cellA}}{k}%
        \leq%
        \distPk{\BallSetA}{\query}{k} + \norm{\query -
           \repres{\cellA}}%
        \leq%
        \distPk{\BallSetA}{\query}{k} + \diameter{\cellA}%
        \leq%
        (1+\eps/4) \distPk{\BallSetA}{\query}{k} .
    \end{align*}
    Therefore, the ball
    $\kbannrepX{\repres{\cellA}}$ satisfies
    \begin{align*}
        \dist{\repres{\cellA}}{\kbannrepX{\repres{\cellA}}}%
        \leq%
        (1+\eps/4)\distPk{\BallSetA}{\repres{\cellA}}{k}%
        \leq%
        (1+\eps/4)^2 \distPk{\BallSetA}{\query}{k}%
        \leq%
        (1+3\eps/4) \distPk{\BallSetA}{\query}{k}.
    \end{align*}
    As such we have that
    \begin{align*}
        \dist{\query}{\kbannrepX{\repres{\cellA}}}%
        \leq%
        \dist{\repres{\cellA}}{\kbannrepX{\repres{\cellA}}} +
        \norm{\query - \repres{\cellA}}%
        \leq%
        \pth{(1+3\eps/4) + \eps/4}\distPk{\BallSetA}{\query}{k}%
        \leq%
        (1+\eps) \distPk{\BallSetA}{\query}{k}.
    \end{align*}
    
    Similarly, using the Lipschitz property, we can argue that,
    $\dist{\query}{\kbannrepX{\repres{\cellA}}} \geq (1-\eps)
    \distPk{\BallSetA}{\query}{k}$, and therefore we have, $(1-\eps)
    \distPk{\BallSetA}{\query}{k} \leq
    \dist{\query}{\kbannrepX{\repres{\cellA}}} \leq (1+\eps)
    \distPk{\BallSetA}{\query}{k}$, and the required guarantees are
    satisfied.
\end{proof}

\begin{lemma}
    \lemlab{anchor:cluster:exist}
    For any point $\query \in \Re^d$ there is a quorum ball $\ballC_i
    = \ball{\ctrC_i}{\radC_i}$ such that
    \begin{inparaenum}[(A)]
        \item $\ballC_i$ intersects
        $\ball{\query}{\distPk{\BallSetA}{\query}{k}}$,
        \item $\radC_i \leq 3 \constC \distPk{\BallSetA}{\query}{k}$,
        and
        \item $\norm{\query - \ctrC_i} \leq 4 \constC
        \distPk{\BallSetA}{\query}{k}$.
    \end{inparaenum}
\end{lemma}
\begin{proof}
    By assumption, $k > 2 \cDim$, and so by \lemref{int:large} among
    the $k$ nearest neighbor of $\query$, there are $k - \cDim$ balls
    of radius at most $\distPk{\BallSetA}{\query}{k}$.  Let
    $\BallSetA'$ denote the set of these balls.  Among the indices
    $1,\ldots, \nclust$, let $i$ be the minimum index such that one of
    these $k-\cDim$ balls is completely covered by the quorum cluster
    $\ballC_i = \ball{\ctrC_i}{\radC_i}$.  Since
    $\ball{\query}{\distPk{\BallSetA}{\query}{k}}$ intersects the ball
    while $\ballC_i$ completely contains it, clearly $\ballC_i$
    intersects $\ball{\query}{\distPk{\BallSetA}{\query}{k}}$.  Now
    consider the time $\ballC_i$ was constructed, i.e, the $i$\th
    iteration of the quorum clustering algorithm. At this time, by
    assumption, all of $\BallSetA'$ was available, i.e., none of its
    balls were assigned to earlier quorum clusters. The ball
    $\ball{\query}{3 \distPk{\BallSetA}{\query}{k}}$ contains $k -
    \cDim$ unused balls and touches $k$ balls from $\BallSetA$, as
    such the smallest such ball had radius at most $3
    \distPk{\BallSetA}{\query}{k}$. By the guarantee on quorum
    clustering, $\radC_i \leq 3 \constC
    \distPk{\BallSetA}{\query}{k}$.  As for the last part, as the
    balls $\ball{\query}{\distPk{\BallSetA}{\query}{k}}$ and $\ballC_i
    = \ball{\ctrC_i}{\radC_i}$ intersect and $\radC_i \leq 3 \constC
    \distPk{\BallSetA}{\query}{k}$, we have by the triangle inequality
    that $\norm{\query - \ctrC_i} \leq (1 + 3 \constC)
    \distPk{\BallSetA}{\query}{k} \leq 4 \constC
    \distPk{\BallSetA}{\query}{k}$, as $\constC \geq 1$.
\end{proof}

\begin{definition}
    For a given query point, any quorum cluster that satisfies the
    conditions of \lemref{anchor:cluster:exist} is defined to be an
    \emphi{anchor cluster}. By \lemref{anchor:cluster:exist} an anchor
    cluster always exists.
\end{definition}

\begin{lemma}
    \lemlab{q:cluster:right:size}
    Suppose that among the quorum cluster balls $\ballC_1, \ldots,
    \ballC_\nclust$, there is some ball $\ballC_i =
    \ball{\ctrC_i}{\radC_i}$ which satisfies that $\norm{\query -
       \ctrC_i} \leq 8 \constC \distPk{\BallSetA}{\query}{k}$ and
    $\eps \distPk{\BallSetA}{\query}{k} / 4 \leq \radC_i \leq 8
    \constC \distPk{\BallSetA}{\query}{k}$ then the output of the
    algorithm is correct.
\end{lemma}

\begin{proof}
    We have
    \[
    \frac{32 \constC \radC_i}{\eps} \geq \frac{32 \constC \pth{\eps
          \distPk{\BallSetA}{\query}{k}/4}}{\eps} = 8 \constC
    \distPk{\BallSetA}{\query}{k} \geq \norm{\query - \ctrC_i}.
    \]
    Thus, by construction, the expanded environ of the quorum cluster
    $\ball{\ctrC_i}{\radC_i}$ contains the query point, see
    \Eqrefpage{clusters:around:q}. Let $j$ be the smallest
    non-negative integer such that $2^j \radC_i \geq
    \dist{\query}{\ctrC_i}$. We have that, $2^j \radC_i \leq \max
    (\radC_i, 2 \dist{\query}{\ctrC_i})$.  As such, if $\cellA$ is the
    smallest cell in $\CellSetAVD$ containing the query point
    $\query$, then
    \begin{align*}
        \diameter{\cellA}%
        &\leq %
        \frac{\eps}{\constA} 2^{j+1} \radC_i%
        \leq %
        \frac{\eps}{\constA} \cdot \max \pth{ 2 \radC_i,4
           \dist{\query}{\ctrC_i} }%
        \leq %
        \frac{\eps}{\constA} \cdot \max \pth{ 16 \constC
           \distPk{\BallSetA}{\query}{k}, 32 \constC
           \distPk{\BallSetA}{\query}{k} \MakeBig } %
        \\
        &\leq%
        \frac{\eps}{8} \distPk{\BallSetA}{\query}{k}, %
    \end{align*}
    by \Eqrefpage{clusters:around:q}, and if $\constA \geq 256
    \constC$.  Now, by \lemref{triv:ub} we have that $\estX \geq
    \distPk{\BallSetA}{\query}{k}$, so $\diameter{\cellA} \leq
    (\eps/8) \estX$. Therefore, the algorithm returns
    $\kbannrepX{\repres{\cellA}}$ as the $(1 \pm \eps)$-approximate
    $k$\th-nearest neighbor, but then by \lemref{small:cell} it is a
    correct answer.
\end{proof}

\begin{lemma}
    \lemlab{corr:final}%
    The query algorithm always outputs a correct approximate answer,
    i.e., the output ball $\ballA$ satisfies $(1-\eps)
    \distPk{\BallSetA}{\query}{k} \leq \dist{\query}{\ballA} \leq
    (1+\eps) \distPk{\BallSetA}{\query}{k} $.
\end{lemma}
\begin{proof}
    Suppose that among the quorum cluster balls $\ballC_1 =
    \ball{\ctrC_1}{\radC_1}, \ldots, \ballC_\nclust =
    \ball{\ctrC_\nclust}{\radC_\nclust}$, there is some ball
    $\ballC_i$ such that $\norm{\query - \ctrC_i} \leq 8 \constC
    \distPk{\BallSetA}{\query}{k}$ and $(\eps/4)
    \distPk{\BallSetA}{\query}{k} \leq \radC_i \leq 8 \constC
    \distPk{\BallSetA}{\query}{k}$, then by
    \lemref{q:cluster:right:size} the algorithm returns a valid
    approximate answer. Assume this condition is not satisfied. Let
    the anchor cluster be $\ballC = \ball{\ctrC}{\radC}$.  Since the
    anchor cluster satisfies $\norm{\query - \ctrC} \leq 4 \constC
    \distPk{\BallSetA}{\query}{k}$ and $\radC \leq 3 \constC
    \distPk{\BallSetA}{\query}{k}$, it must be the case that, $\radC <
    (\eps/4) \distPk{\BallSetA}{\query}{k}$. Since the anchor cluster
    intersects $\ball{\query}{\distPk{\BallSetA}{\query}{k}}$, we have
    that $\norm{\query - \ctrC} \leq (1+\eps/4)
    \distPk{\BallSetA}{\query}{k}$.  Thus, $\normP{\mapped{\query} -
       \mapped{\ballC}} = \norm{\query - \ctrC} + \radC \leq
    (1+\eps/2) \distPk{\BallSetA}{\query}{k}$.  Let $\cellA$ be the
    smallest cell in which $\query$ is located. Now consider the point
    $\repX{\cellA} \in \mapped{\QClustBalls}$.  Suppose it corresponds
    to the cluster $\ballC_j$, i.e., $\mapped{\ballC_j} =
    \repX{\cellA}$. Since $\repX{\cellA}$ is a $(1+\eps/8)$-\ANN to
    $\query$ among the points of $\mapped{\QClustBalls}$,
    $\normP{\mapped{\query} - \repX{\cellA}} \leq (1+\eps/8)
    \normP{\mapped{\query} - \mapped{\ballC}} \leq
    (1+\eps/8)(1+\eps/2)\distPk{\BallSetA}{\query}{k} \leq (1+\eps)
    \distPk{\BallSetA}{\query}{k} \leq 2 \distPk{\BallSetA}{\query}{k}
    \leq 8 \constC \distPk{\BallSetA}{\query}{k} $. It follows that,
    $\norm{\query - \ctrC_j} \leq 8 \constC
    \distPk{\BallSetA}{\query}{k}$, and $\radC_j \leq 8 \constC
    \distPk{\BallSetA}{\query}{k}$.  By our assumption, it must be the
    case that, $\radC_j < (\eps/4) \distPk{\BallSetA}{\query}{k}$.
    Now, there are two cases. Suppose that, $\diameter{\cellA} \leq
    (\eps/8) \estX$. Then, since we have $\estX \leq
    \normP{\mapped{\query} - \repX{\cellA}}$ so $\estX \leq 2
    \distPk{\BallSetA}{\query}{k}$. As such, $\diameter{\cellA} \leq
    (\eps/4) \distPk{\BallSetA}{\query}{k}$.  In this case we return
    $\kbannrepX{\cellA}$ by the algorithm, but the result is correct
    by \lemref{small:cell}. On the other hand, if we return
    $\ballrepX{\cellA}$, it is easy to see that
    $\dist{\query}{\ballrepX{\cellA}} \leq \norm{\query - \ctrC_j} +
    \radC_j \leq (1+\eps) \distPk{\BallSetA}{\query}{k} $. Also, as
    $\ballrepX{\cellA}$ lies completely inside $\ballC_j$ it follows
    by \obsref{ball:contained}, that $\dist{\query}{\ballrepX{\cellA}}
    \geq \dist{\query}{\ballC_j} \geq \norm{\query - \ctrC_j} -
    \radC_j \geq (\norm{\query - \ctrC_j} + \radC_j) - 2\radC_j \geq
    \distPk{\BallSetA}{\query}{k} - (\eps/2)
    \distPk{\BallSetA}{\query}{k} \geq
    (1-\eps/2)\distPk{\BallSetA}{\query}{k}$, where the second last
    inequality follows by \lemref{triv:ub}.
\end{proof}

\subsection{The result}

\begin{theorem}%
    \thmlab{avd:main}%
    Given a set $\BallSetA$ of $n$ disjoint balls in $\Re^d$, a number
    $k$, with $1 \leq k \leq n$, and $\eps \in (0,1)$, one can
    preprocess $\BallSetA$, in $\ds O \pth{ n \log n + \frac{n}{k }
       C_\eps \log n + \frac{n}{k } C_\eps'}$ time, where $C_\eps =
    O\pth{ \eps^{-d}\log {\eps}^{-1} }$ and $C_\eps' = O\pth{
       \eps^{-2d}\log {\eps}^{-1} }$. The space used by the
    data structure is $O( C_\eps n/k)$. Given a query point $\query$,
    this data structure outputs a ball $\ballA \in \BallSetA$ in $\ds
    O\pth[]{\log \frac{n}{k \eps }}$ time, such that $(1-\eps)
    \distPk{\BallSetA}{\query}{k} \leq \dist{\query}{\ballA} \leq
    (1+\eps)\distPk{\BallSetA}{\query}{k}$.
\end{theorem}

\begin{proof}
    If $k \leq 2\cDim$ then \thmref{k:at:query:main} provides the
    desired result.
    For $k > 2 \cDim$, the correctness was proved in 
    \lemref{corr:final}. We only need to
    bound the construction time and space as well as the query time.
    Computing the quorum clustering takes time $O(n \log n)$ by
    \lemref{q:clustering}.  Observe that $\cardin{\CellSetA} =
    O\pth[]{\frac{n}{k\eps^d}\log\frac{1}{\eps}}$. From the
    construction of Arya and Malamatos \cite{am-lsavd-02}, we have
    $\cardin{\CellSetC} = O\pth[]{\frac{n}{k \eps^{d}}\log
       \frac{1}{\eps}}$ (note, that since we clip the construction to
    a hyperplane, we get $1/\eps^d$ in the bound and not
    $1/\eps^{d+1}$). A careful implementation of this stage takes time
    $O\pth{ n \log n + \cardin{\CellSetAVD}\pth{\log n +
          \frac{1}{\eps^{d-1}}}}$. Overlaying the two compressed
    quadtrees representing them takes linear time in their size, that
    is $O\pth{ \cardin{\CellSetA} + \cardin{\CellSetC} }$.
    
    The most expensive step is to perform the $(1 \pm
    \eps/4)$-approximate $k$\th-nearest neighbor query for each cell
    in the resulting decomposition of $\CellSetAVD$, see
    \Eqrefpage{in:cell} (i.e., computing $\adknn{\repres{\cellA}}$ for
    each cell $\cellA \in \CellSetAVD$). Using the data structure of
    \secref{k:at:query} (see \thmref{k:at:query:main}) each query
    takes $O\pth{ \log n + 1/\eps^{d}}$ time.
    \begin{align*}
        O\pth{ n \log n + \cardin{\CellSetAVD}\pth{\log n +
              \frac{1}{\eps^{d}}}}%
        =%
        O \pth{%
           n \log n + %
           \frac{n}{k \eps^{d}}\log \frac{1}{\eps}%
           \log n +%
           \frac{n}{k \eps^{2d}}\log \frac{1}{\eps}%
        }
    \end{align*}
    time, and this bounds the overall construction time.
    
    The query algorithm is a point location query followed by an
    $O(1)$ time computation and takes time $O\pth[]{\log
       \pth{\frac{n}{k \eps}}}$.
\end{proof}

Note that the space decomposition generated by \thmref{avd:main} can
be interpreted as a space decomposition of complexity $O( C_\eps
n/k)$, where every cell has two input balls associated with it, which
are the candidates to be the desired $(k,\eps)$-ANN. That is,
\thmref{avd:main} computes a $(k.\eps)$-\AVD of the input balls.

\section{Conclusions}
\seclab{conclusions}

In this paper, we presented a generalization of the usual $(1 \pm
\eps)$-approximate $k$\th-nearest neighbor problem in $\Re^d$, where
the input are balls of arbitrary radius, while the query is a point.
We first presented a data structure that takes $O(n)$ space, and the
query time is $O(\log n + \eps^{-d})$. Here, both $k$ and $\eps$ could
be supplied at query time. Next we presented an $(k,\eps)$-\AVD taking
$O(n / k)$ space. Thus showing, surprisingly, that the problem
can be solved in sublinear space if $k$ is sufficiently large.





\bibliographystyle{plain}%
\bibliography{shortcuts,geometry}%

\end{document}